\documentclass[11pt]{article}

\usepackage[final]{acl}

\usepackage{times}
\usepackage{latexsym}
\usepackage[T1]{fontenc}
\usepackage[utf8]{inputenc}
\usepackage{microtype}
\usepackage{inconsolata}
\usepackage{graphicx}
\usepackage{amsmath}
\usepackage{amssymb}
\usepackage{amsthm}
\usepackage{booktabs}
\usepackage{enumitem}

\newtheorem{definition}{Definition}

\newtheorem{theorem}{Theorem}
\newtheorem{proposition}{Proposition}

\newcommand{\Hq}{\nabla^2 C(q)}
\DeclareMathOperator*{\argmax}{arg\,max}

\title{Trade-Adaptive Aggregation of Probabilistic Financial Forecasts}

\author{%
  \textbf{Yankai Chen}\textsuperscript{1,2},
  \textbf{Rassul Magauin}\textsuperscript{2},
  \textbf{Bowei He}\textsuperscript{1,2},
  \textbf{Anuar Aimoldin}\textsuperscript{2},
\\
  \textbf{Sirui Song}\textsuperscript{1},
  \textbf{Bin Xiao}\textsuperscript{3},
  \textbf{Zangir Iklassov}\textsuperscript{2},
  \textbf{Xue Liu}\textsuperscript{1,2}
\\
  \textsuperscript{1}McGill University \quad
  \textsuperscript{2}MBZUAI \quad
  \textsuperscript{3}The Hong Kong Polytechnic University
}

\newcommand{\newtext}[1]{#1}

\begin{document}
\maketitle

\begin{abstract}
Financial NLP systems produce probabilistic forecasts from news, reports, and filings. Prediction markets can aggregate these forecasts sequentially, but their fees must reward information without overcharging low-risk updates. Existing quadratic-fee mechanisms use a state-blind bound, while a local-curvature envelope remains conservative because it prices every trade at the largest permitted span. We introduce \emph{SpanPM}, a prediction-market mechanism that sets the local-curvature multiplier from each trade's realized payoff spread. Its fee dominates exact Bregman exposure trade by trade, preserves no arbitrage, information incorporation, expressiveness, and bounded worst-case loss, and yields a tighter overcharge factor approaching one as trade span vanishes. Repeated global best responses converge to a common belief and become full Newton steps locally, giving quadratic rather than damped-linear convergence. We implement a deterministic bounded one-dimensional multi-basin search, audited against a dense grid. Across paired synthetic experiments, SpanPM improves 20-round consensus error by several orders of magnitude over a fixed-envelope local baseline under the same hard cap. With evolving beliefs, it preserves 96--97\% of the trader surplus achieved with exact Bregman fees while cutting excess fees by 94\% relative to the global quadratic mechanism. These results establish a trade-adaptive prediction-market mechanism for sequential aggregation of probabilistic financial forecasts.
\end{abstract}

\section{Introduction}
Financial NLP systems turn news, reports, and filings into probabilistic predictions such as sentiment labels and event outcomes\newtext{~\cite{araci2019finbert,xie2023pixiu,wang2026finauditing,wang2026convfinre,zhang-etal-2026-finreporting}}.
In practice, several models, \newtext{agents,} or analysts may issue forecasts sequentially, with different confidence and information\newtext{~\cite{peng2026herculean,xie2026finmmevaltask3,gong2026shijianbench}}.
Prediction markets offer a transparent aggregation layer: each forecast acts as a trader belief, while the market price summarizes the beliefs incorporated so far~\cite{hanson2003combinatorial,wolfers2004prediction}.
This paper studies the mechanism after a forecaster has produced a probability vector; it does not propose a new text encoder.

An automated market maker (AMM) quotes prices and accepts trades even without a matched counterparty.
In a cost-function market maker, a convex cost function determines the price and the payment for a trade~\cite{abernethy2013efficient}.
The payment can be decomposed into the current-price value plus a Bregman fee that compensates the market maker for incremental exposure.
SQPM replaces this exact fee with a global quadratic upper bound~\cite{nueve2025smooth}.
The bound supplies pathwise protection but relaxes one-shot incentive compatibility: information is incorporated through repeated, smaller trades.

\textbf{The global fee can be badly calibrated to local risk.}
SQPM charges the same quadratic fee for an equal-sized trade at every market state.
When the price is already confident, however, a trade may change little payoff risk and have a very small Bregman fee.
The state-independent surcharge does not shrink with that risk and can therefore dominate the value of the update.
For sequential financial forecasts, this matters because a costly aggregation rule may suppress small corrections precisely when the consensus is confident.

\begin{figure}[t]
\centering
\includegraphics[width=0.88\columnwidth]{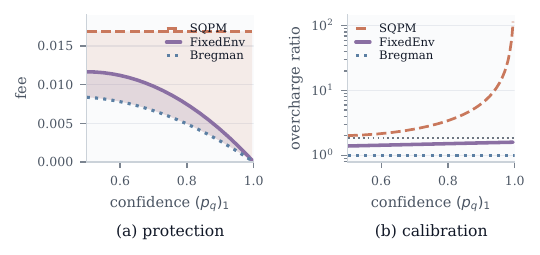}
\caption{Fee calibration for a fixed trade as market confidence increases. \textbf{(a)} A local-curvature fee's excess above the Bregman fee shrinks with local risk, whereas SQPM's global fee does not. \textbf{(b)} Its fee-to-Bregman ratio remains uniformly bounded. SpanPM retains these state-adaptive benefits and further adapts the multiplier to realized trade spread.}
\label{fig:overcharge-1d}
\end{figure}

Local-curvature fees improve calibration, but a fixed-envelope construction still inflates \emph{every} trade by the worst curvature drift allowed at the cap. This is unnecessarily conservative near consensus, where optimal trades are small. We call the fixed-envelope baseline \emph{FixedEnv} and our realized-span mechanism \emph{SpanPM}. SpanPM scales the local Hessian quadratic by the curvature drift implied by the trade's \emph{realized} payoff spread, while retaining one hard maximum spread for safety. Thus both the state and the trade determine the charge.

The construction exposes a useful design trade-off.
For a fixed price rule, we show that the Bregman fee is the unique differentiable fee preserving one-shot incentive compatibility.
Stronger per-trade protection must therefore relax that property and rely on a sequence of trades.
Within this regime, SpanPM preserves no arbitrage, expressiveness, information incorporation, and bounded worst-case loss. Its per-trade overcharge factor is \(\rho(\kappa s(r))\), strictly tighter than using the hard cap whenever the realized spread is smaller. The vanishing inflation also changes sequential aggregation: SQPM induces gradient-like updates, the fixed-envelope local baseline induces damped Newton-like updates, and SpanPM approaches a full Newton step as consensus makes trades small. We prove convergence and a local quadratic rate.

Our contributions are:
\begin{itemize}[leftmargin=*,topsep=2pt,itemsep=1pt,parsep=0pt,partopsep=0pt]
\item We characterize the one-shot incentive-compatibility constraint in fixed price-plus-fee markets and formalize the calibration failure of a global quadratic fee.
\item We propose SpanPM, a trade-span-adaptive local fee with per-trade safety, information incorporation, and a realized rather than worst-cap calibration factor.
\item We prove global consensus convergence and local quadratic convergence, and reduce each nonquadratic response problem to a bounded one-dimensional search, evaluated by a deterministic multi-basin method and audited on a dense grid.
\item In paired simulations and solver audits, we compare SpanPM with DCFMM, SQPM, and a fixed-envelope local baseline, including convergence, fees, surplus, cap sensitivity, and numerical safety checks.
\end{itemize}

The experiments deliberately isolate the aggregation mechanism using synthetic probability forecasts\newtext{~\cite{liu2026finsimsurvey}}. This makes the mechanism comparison controlled, while leaving end-to-end evaluation with financial text models as an important next step.

\section{Preliminaries}

\textbf{Market and cost functions.}
We study a prediction market with \(d\) mutually exclusive and exhaustive outcomes, trading Arrow--Debreu securities \cite{arrow1964role}: a trader who purchases bundle \(r\in\mathbb{R}^d\) pays the market maker for it and receives \(r_i\) if outcome \(i\) is realized.
A valid set of prices is therefore a probability distribution over these \(d\) outcomes. We write the set of all such distributions as the simplex \(\Delta_d=\{p\in\mathbb{R}^d_+:\langle \mathbf{1},p\rangle=1\}\), where \(\mathbf{1}\in\mathbb{R}^d\) is the all-ones vector. Market prices must lie in \(\Delta_d\).
A cost-function market maker \cite{abernethy2013efficient} keeps a state \(q\in\mathbb{R}^d\) recording cumulative inventory sold, and a trade \(r\) moves the state from $q$ to \(q+r\).
The market is specified by a convex, increasing, translation-invariant cost function \(C:\mathbb{R}^d\to\mathbb{R}\) (standard technical assumptions in the appendix) whose gradient \(p_q=\nabla C(q)\in\Delta_d\) is the instantaneous price, with price range dense in \(\Delta_d\) (able to get arbitrarily close to any valid price).
Intuitively, this price is the market's current belief about how likely each outcome is.
Translation invariance makes \(\nabla^2C(q)\mathbf{1}=0\), so the Hessian is singular along the riskless direction \(\mathbf{1}\) and the market's meaningful local geometry lives on the zero-sum subspace \(E=\{v\in\mathbb{R}^d:\langle \mathbf{1},v\rangle=0\}\), i.e., the geometry our proposed fee will exploit.
% Intuitively, shifting the inventory of every outcome by the same amount changes no relative price and creates no risk, so all of the curvature that matters for fee design comes from trades that shift confidence \emph{between} outcomes, not trades along this riskless direction.

\textbf{Bregman divergence.}
For a differentiable convex function \(f\), the Bregman divergence is \(D_f(x,y)=f(x)-f(y)-\langle \nabla f(y),x-y\rangle\ge 0\) \cite{bregman1967relaxation}, the error from approximating \(f(x)\) by the first-order Taylor expansion at \(y\).
When \(f\) is twice differentiable, it also admits the curvature representation:
\begin{equation}
    D_f(y+r,y)
    =
    \int_0^1
        (1-\theta)\,
        \langle r,\nabla^2 f(y+\theta r)r\rangle
    \,d\theta ,
\end{equation}
i.e., the accumulated curvature along the trade segment -- the starting point for our fee design.
A function is \(L\)-smooth if \(D_f(x,y)\le \frac{L}{2}\|x-y\|^2\) for all \(x,y\); this global bound is exactly what SQPM uses as its fee.

\textbf{Price-plus-fee markets.}
Every payment decomposes as
\begin{equation}
    \mathrm{Pay}(q,r)
    =
    \langle \nabla C(q),r\rangle
    +
    \mathrm{Fee}(q,r),
\end{equation}
where the first term \(\langle \nabla C(q),r\rangle\) values the trade at the current instantaneous price \(p_q=\nabla C(q)\) (price times quantity, summed over the \(d\) securities), and \(\mathrm{Fee}(q,r)=o(\|r\|)\) as \(r\to 0\), so that the instantaneous price is unaffected by the fee \cite{nueve2025smooth}.
Intuitively, the trader pays the current price plus a fee covering the trade's added risk.
The classical Duality-based Cost-function Market Maker (DCFMM) \cite{abernethy2013efficient} charges the exact cost difference \(C(q+r)-C(q)\); by the definition of Bregman divergence, its fee is exactly \(\mathrm{Fee}_D(q,r)=D_C(q+r,q)\).
Smooth Quadratic Prediction Markets (SQPM) instead charges the global quadratic surcharge \cite{nueve2025smooth}, using the smoothness constant \(L\) of \(C\) introduced above:
\begin{equation}
    \mathrm{Fee}_L(q,r)
    =
    \frac{L}{2}\|r\|^2 ,
\end{equation}
which dominates the Bregman fee whenever \(C\) is \(L\)-smooth but depends only on \(\|r\|\), not on the market state \(q\).

\textbf{Market-design guarantees.}
We use the standard requirements for prediction-market mechanisms: (1) no arbitrage (no sequence of trades yields a risk-free profit), (2) expressiveness (every belief in the price range is reachable by some trade), (3) information incorporation (buying the same bundle again after it moves the market does not become cheaper), (4) bounded worst-case loss (the market maker cannot lose more than a fixed amount, however the outcome resolves), and (5) incentive compatibility (formal definitions in the appendix).
The DCFMM satisfies (1)--(4) and achieves the strongest form of (5), \emph{one-shot} incentive compatibility \cite{abernethy2013efficient}: a trader can move the market price to their belief in a single trade.
SQPM also satisfies (1)--(4) but relaxes (5) to an \emph{incremental} form: repeated trading by like-minded traders still drives prices to consensus, just not in a single trade \cite{nueve2025smooth}.
This distinction is central to our paper: once relaxed, the fee no longer only protects the market maker, but also determines how quickly prices move toward traders' beliefs.

\textbf{Softmax running example.}
Our main concrete example is the softmax cost function, also known as the Logarithmic Market Scoring Rule (LMSR), a standard prediction-market mechanism \cite{hanson2003combinatorial}.
\begin{equation}
    C(q)=\frac{1}{L}\log\sum_{i=1}^d e^{Lq_i},
    \quad
    (p_q)_i=\frac{e^{Lq_i}}{\sum_{j=1}^d e^{Lq_j}},
\end{equation}
where \(L>0\).
The price vector lies in \(\operatorname{relint}(\Delta_d)=\{p\in\Delta_d:p_i>0\text{ for all }i\}\): softmax prices never fully rule out any outcome, only approach that boundary as the market gets more confident, and the closure of the price range is the full simplex \(\Delta_d\).
The Hessian, and the resulting quadratic form for any trade \(r\), are
\begin{align}
    \nabla^2 C(q)
    &=
    L\left(\operatorname{diag}(p_q)-p_qp_q^\top\right), \notag\\
    \langle r,\nabla^2 C(q)r\rangle
    &=
    L\,\operatorname{Var}_{p_q}(r).
\end{align}
Thus, in the softmax market, local curvature is proportional to the payoff variance of the trade under the current belief, large when the payoff is uncertain and small when it is nearly deterministic. This is exactly the distinction a state-independent quadratic fee cannot capture.

\section{Realized-Span Fee Design}

\label{sec:ca-fees}

\textbf{Overview.}
This section develops SpanPM's fee from first principles.
We begin with a rigidity result: for a fixed price rule, the Bregman fee is the unique differentiable fee that preserves one-shot incentive compatibility. This explains why any mechanism offering stronger per-trade protection must relax it.
We then identify a sufficient condition for safe replacement fees and derive a fixed-envelope local-curvature baseline. Finally, we replace its worst-cap multiplier by the realized trade span. The resulting fee retains the standard guarantees while tightening calibration and changing the local learning rate.

\subsection{Fee Design Beyond One-Shot Incentive Compatibility}

\textbf{Why the fee is the design object.}
In a price-plus-fee market, the cost function \(C\) fixes the instantaneous price rule \(p_q=\nabla C(q)\). The fee then determines what the trader pays beyond the current-price value of the trade.
Thus, once the price rule is fixed, fee design controls the tradeoff between trader incentives, market-maker protection, and information incorporation.
Recall from \emph{Preliminaries} that the DCFMM's Bregman fee achieves one-shot incentive compatibility.
The following proposition shows that this property leaves no freedom to choose a different differentiable fee.

\begin{proposition}[Uniqueness of the Bregman fee]
\label{prop:bregman-unique}
Fix a differentiable, translation-invariant cost function \(C\) whose price map is injective on the zero-sum subspace \(E\).
Consider a differentiable price-plus-fee market with payment
\begin{equation}
    \mathrm{Pay}(q,r)
    =
    \langle \nabla C(q),r\rangle
    +
    \mathrm{Fee}(q,r),
\end{equation}
where \(\mathrm{Fee}(q,0)=0\).
Assume also that \(\mathrm{Fee}(q,\cdot)\) depends on a trade only through its risky component, i.e., \(\mathrm{Fee}(q,r+\alpha\mathbf{1})=\mathrm{Fee}(q,r)\) for all \(\alpha\in\mathbb{R}\).
A riskless overlay changes every outcome's payoff by the same amount and therefore creates no risk.
If the market satisfies one-shot incentive compatibility for every belief in the price range of \(C\), then
\begin{equation}
    \mathrm{Fee}(q,r)=D_C(q+r,q).
\end{equation}
\end{proposition}

The proof is deferred to the appendix.
The proposition formalizes a basic rigidity: under the fixed price rule \(\nabla C\), one-shot incentive compatibility forces the Bregman fee.
Consequently, any mechanism that charges strictly more than the DCFMM fee somewhere must relax one-shot incentive compatibility.
The design goal then becomes preserving market safety and information incorporation while allowing stronger protection for the market maker.

\subsection{A Sufficient Condition for Fee Domination}

\textbf{Dominating the DCFMM fee.}
We next state a general condition under which a replacement fee preserves the standard market-design guarantees.
The fee must dominate the Bregman fee for protection.
It must also remain negligible at infinitesimal scale, so that prices are unchanged.
Finally, it must vary stably with the market state, so that repeated purchases do not become cheaper after the market moves.

\begin{definition}[\(C\)-dominating fee]
\label{def:dominating-fee}
A fee \(\mathrm{Fee}\) is \(C\)-dominating if, for all states \(q\) and every admissible trade \(r\), it satisfies
\begin{align}
    \mathrm{Fee}(q,r)
    &\ge D_C(q+r,q), \label{eq:dominating-protection}\\
    \mathrm{Fee}(q,r)
    &= o(\|r\|)
    \quad \text{as } r\to 0, \label{eq:dominating-price}\\
    \mathrm{Fee}(q,r)-\mathrm{Fee}(q+r,r)
    &\le \langle p_{q+r}-p_q,r\rangle .
    \label{eq:dominating-drift}
\end{align}
\end{definition}

Every trade is admissible for the DCFMM and SQPM; the local-envelope mechanisms below admit only capped trades.
Condition~\eqref{eq:dominating-protection} is the protection condition.
Condition~\eqref{eq:dominating-price} preserves the instantaneous price.
Condition~\eqref{eq:dominating-drift} is the information-incorporation condition: after the market moves in the direction of \(r\), the fee for repeating the same trade cannot fall too much.

\begin{theorem}[Guarantees from dominating fees]
\label{thm:dominating-guarantees}
If \(\mathrm{Fee}\) is \(C\)-dominating, then the corresponding price-plus-fee market has well-defined instantaneous prices, satisfies no arbitrage, preserves expressiveness, and incorporates information.
Moreover, on every trade history, the market maker's collected revenue is at least the revenue of the corresponding DCFMM.
Consequently, the market maker's worst-case loss is no larger than under the DCFMM.
\end{theorem}

The proof is given in the appendix.
This theorem gives the template for the rest of the paper: a fee can improve market-maker protection without preserving one-shot incentive compatibility.
It only needs to dominate the Bregman fee and satisfy the state-stability condition above.

\subsection{Miscalibration of the Global Quadratic Fee}

\textbf{SQPM as global protection.}
SQPM~\cite{nueve2025smooth} follows the domination principle using a global smoothness bound: if \(C\) is \(L\)-smooth, then
\begin{equation}
    D_C(q+r,q)
    \le
    \frac{L}{2}\|r\|^2.
\end{equation}
Thus the SQPM fee \(\mathrm{Fee}_L(q,r)=\frac{L}{2}\|r\|^2\) dominates the DCFMM fee and gives the market maker additional pathwise protection \cite{nueve2025smooth}.

The limitation is calibration.
The SQPM fee depends only on the size of the trade \(r\), not on the current market state \(q\) or belief \(p_q\).
In confident markets, some trades add little incremental payoff risk, so the Bregman fee can become small, but the global quadratic fee does not respond to this reduction in local risk.

\begin{proposition}[Unbounded overcharge of the global quadratic fee]
\label{prop:sqpm-unbounded-overcharge}
For the softmax cost function, there exist fixed nonzero trades \(r\) such that the overcharge ratio is unbounded: for every \(M>0\), there is a market state \(q\) at which
\begin{equation}
    \frac{\mathrm{Fee}_L(q,r)}{D_C(q+r,q)}
    >
    M .
\end{equation}
\end{proposition}

The proof is deferred to the appendix.
Intuitively, consider a market that is already almost certain about one outcome, assigning nearly all probability to it and almost none to the rest.
For a trade whose payoff is already nearly determined under that belief, there is almost no risk left to price: the trade's payoff barely varies across the few outcomes still in doubt.
Recall from \emph{Preliminaries} that the Bregman fee tracks exactly this payoff variance, so it shrinks toward zero as the market approaches certainty.
The SQPM fee, however, depends only on the trade's size \(\|r\|\), not on how confident the market is, so it does not shrink at all.
This is why SQPM can overcharge by an arbitrarily large factor precisely in the states where the market is most confident.

\subsection{Local-Curvature Envelopes}

\textbf{From global smoothness to local curvature.}
This subsection builds a local-curvature envelope in three steps: replace SQPM's global curvature bound with the local curvature at the trade's starting state; control how much that curvature can drift during a single trade; and use this control to bound the true risk by the local curvature, within a known factor.

\begin{figure}[t]
\centering
\includegraphics[width=0.54\columnwidth]{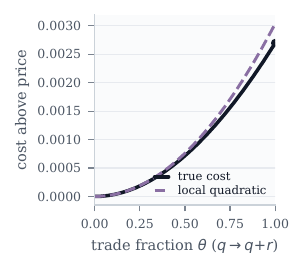}
\caption{How curvature prices risk (details in the appendix). Solid: exact cost above the current price. Dashed: the local-curvature estimate at the trade's starting state. Their close match motivates a local envelope; SpanPM further scales it by realized trade spread.}
\label{fig:curvature-stability}
\end{figure}

Recall from \emph{Preliminaries} that the Bregman fee admits a curvature integral representation: it is the accumulated curvature encountered along the trade segment from \(q\) to \(q+r\).
SQPM upper-bounds this integral by a global quadratic term; a local envelope instead uses the curvature at the current state \(q\).
This is not automatically safe, however: because curvature may increase as the state moves from \(q\) to \(q+r\), a purely local fee proportional to \(\langle r,\nabla^2 C(q)r\rangle\) could understate the risk encountered along the way.
We therefore impose a curvature-stability condition to control how much the Hessian can change during a single trade.
This multiplicative control of Hessian variation is closely related to generalized self-concordance, which supports local and global analyses of Newton-type methods~\cite{sun2019generalized}.

\begin{definition}[Curvature stability]
\label{def:curvature-stability}
Let \(s:\mathbb{R}^d\to\mathbb{R}_+\) be a trade-size gauge, i.e., a nonnegative, positively homogeneous function with \(s(\theta u)=\theta\,s(u)\) for all \(u\in\mathbb{R}^d\) and \(\theta\ge 0\).
A twice differentiable cost function \(C\) satisfies \(\kappa\)-curvature stability with respect to \(s\) if, for all states \(q\) and directions \(u\),
\begin{equation}
    e^{-\kappa s(u)}\nabla^2 C(q)
    \preceq
    \nabla^2 C(q+u)
    \preceq
    e^{\kappa s(u)}\nabla^2 C(q),
\end{equation}
where \(\preceq\) denotes the Loewner order on symmetric matrices.
\end{definition}

Intuitively, this condition rules out cost functions whose local risk profile can swing wildly within a single trade: the curvature measured at the start of a trade stays a faithful description of the curvature encountered throughout, up to a controlled multiplicative factor \(e^{\pm\kappa s(u)}\).
For softmax, we use the spread gauge \(s(r)=\max_i r_i-\min_i r_i\), which ignores riskless shifts in the \(\mathbf{1}\) direction; the softmax cost with parameter \(L\) satisfies curvature stability under this gauge with \(\kappa=2L\) (verified in the appendix).

Curvature stability lets us compare the Bregman fee with the local quadratic form.
Define
\begin{equation}
    \omega^+(a)=\frac{e^a-1-a}{a^2},
    \quad
    \omega^-(a)=\frac{e^{-a}-1+a}{a^2},
\end{equation}
with \(\omega^+(0)=\omega^-(0)=1/2\), so both recover the plain quadratic form when curvature does not change at all along the trade; away from \(a=0\), \(\omega^+\) and \(\omega^-\) bound how far the true, integrated risk can drift above or below that quadratic estimate.
For \(a=\kappa s(r)\), curvature stability implies
\begin{equation}
    \begin{aligned}
        \omega^-(a)\langle r,\nabla^2 C(q)r\rangle
        &\le D_C(q+r,q),\\
        D_C(q+r,q)
        &\le \omega^+(a)\langle r,\nabla^2 C(q)r\rangle .
    \end{aligned}
\end{equation}
This local two-sided bound, proved in the appendix, is the technical bridge from the Bregman fee to a curvature-adaptive envelope.

\textbf{A fixed-envelope baseline.}
We impose a per-trade cap \(s(r)\le\tau\) and define \(A=\kappa\tau\); we call a trade \(r\) admissible if \(s(r)\le\tau\).
The market enforces this cap directly, as made precise by the constrained trader optimization in the next section.
The fixed-envelope local fee is
\begin{equation}
    \mathrm{Fee}_A(q,r)
    =
    \omega^+(A)\langle r,\nabla^2 C(q)r\rangle .
\end{equation}
The multiplier \(\omega^+(A)\) compensates for the largest curvature increase that can occur along any admissible trade, so the fee is local in the current market state, but still large enough to dominate the exact Bregman fee (Figure~\ref{fig:curvature-stability}).

Due to the page limit, we prove in the appendix that this fixed-envelope fee satisfies \emph{safety and bounded overcharge} (Theorem~\ref{thm:ca-safety-overcharge}): it always dominates the Bregman fee, protecting the market maker, and its excess charge is bounded by a constant factor \(\rho(A)\approx 1.83\) that does not depend on the market state; unlike SQPM, whose overcharge is unbounded (Figure~\ref{fig:overcharge-simplex}, in the appendix).
We also prove \emph{information incorporation} (Theorem~\ref{thm:ca-information-incorporation}): restricting trades to \(A\le A^\star\) ensures that the local fee does not fall too quickly after a trade.

\subsection{SpanPM: Pricing the Realized Trade Span}

The fixed-envelope multiplier \(\omega^+(A)\) is chosen for a boundary trade even when the realized trade is much smaller. SpanPM removes this slack. Under a hard cap \(s(r)\le\bar\tau\), let \(a(r)=\kappa s(r)\) and define
\begin{equation}
 \mathrm{Fee}_{\mathrm{span}}(q,r)
 =\omega^+(a(r))\langle r,\nabla^2C(q)r\rangle .
 \label{eq:span-fee}
\end{equation}
The cap remains an enforceable safety envelope, but it no longer determines the multiplier of every interior trade.

\begin{theorem}[Trade-wise guarantees]
\label{thm:span-guarantees}
Suppose \(C\) is \(\kappa\)-curvature stable and \(\bar A=\kappa\bar\tau\le A^\star\). For every admissible trade, write \(D=D_C(q+r,q)\). Then
\begin{equation}
 D\le \mathrm{Fee}_{\mathrm{span}}(q,r)
 \le \rho(a(r))D\le \rho(\bar A)D.
\end{equation}
Moreover, the fee is \(C\)-dominating. Hence SpanPM preserves instantaneous prices, no arbitrage, expressiveness, information incorporation, and worst-case loss no larger than the corresponding DCFMM.
\end{theorem}

The key distinction is that the first factor is realized: \(\rho(a(r))\to1\) as \(s(r)\to0\). For the drift condition, the same trade \(r\) has the same multiplier before and after the state change, so the fixed-envelope argument applies with \(A\) replaced by \(a(r)\). Full proofs are in Appendix~\ref{app:fee-design}.

\section{Price Discovery Dynamics}
\label{sec:dynamics}

The previous section showed that SpanPM's fee is safe and risk-calibrated.
This section asks how it affects learning in the market.
Once one-shot incentive compatibility is relaxed, prices do not generally jump to a trader's belief in one trade.
Instead, information is incorporated through a sequence of trades, and we show that the fee determines the geometry of this sequence.
SQPM induces gradient-like dynamics, taking a step of fixed size at every state.
The fixed-envelope baseline induces damped Newton-like dynamics. SpanPM additionally removes the fixed damping as trades shrink, producing full Newton behavior near consensus.

\subsection{Fees as Trading Geometry}

\textbf{Trader best responses.}
Consider a sequence of traders with a common belief \(\mu\in\Delta_d\).
At state \(q\), a trader chooses a trade \(r\) to maximize expected profit; for a price-plus-fee market, this is equivalent to
\begin{align}
    &\max_r
    \left\{
        \langle \mu,r\rangle
        -
        \mathrm{Pay}(q,r)
    \right\}
    \notag\\
    &\quad\Longleftrightarrow\quad
    \min_r
    \left\{
        \langle \nabla C(q)-\mu,r\rangle
        +
        \mathrm{Fee}(q,r)
    \right\}.
\end{align}
This expression links market design to optimization: the linear term points in the direction of disagreement between the current market price and the trader's belief, while the fee regularizes the size and geometry of the trade.
In words, the trader wants to push the price toward their own belief, but the fee makes doing so more expensive as the push grows larger or riskier, so their optimal trade balances the perceived pricing error against the cost of correcting it.
Changing the fee therefore changes the optimization method that traders collectively implement.

\subsection{From Damped to Full Newton Updates}

\textbf{The fixed-envelope fee induces damped Newton dynamics.}
Both mechanisms can be described as optimizing the same convex potential, \(\Phi_\mu(q)=C(q)-\langle \mu,q\rangle\), whose gradient is \(\nabla C(q)-\mu\).
SQPM makes traders collectively implement steepest gradient descent on this potential~\cite{nueve2025smooth}: a fixed-size step that uses the same global Euclidean geometry at every market state.
The local-curvature fee replaces this global Euclidean geometry with the Hessian geometry of the cost function.
The trader solves
\begin{equation}
    \min_{r:\,s(r)\le\tau}
    \left\{
        \langle p_q-\mu,r\rangle
        +
        \omega^+(A)\langle r,\Hq r\rangle
    \right\}.
\end{equation}
This is a local quadratic model of \(\Phi_\mu\), measured in the curvature of \(C\) at the current state.
When the cap is inactive, the first-order condition gives \(2\omega^+(A)\nabla^2 C(q)r=\mu-\nabla C(q)\); thus, up to damping and the singular riskless direction, the update is a Newton step:
\begin{equation}
    r
    =
    \frac{1}{2\omega^+(A)}
    \nabla^2 C(q)^{\dagger}
    \bigl(\mu-\nabla C(q)\bigr),
\end{equation}
where \(\nabla^2 C(q)^{\dagger}\) denotes the Moore--Penrose pseudoinverse on the non-riskless subspace.

The difference from SQPM is preconditioning.
SQPM uses the same fixed step size at every market state, while the local-curvature fee adapts the step to the market's geometry.
When the market is confident and curvature is small, the local-curvature rule takes a larger step, whereas SQPM's fee still restrains the trade as if the market were uncertain.
This adaptive step size leads to faster price discovery, which we show formally below.

\subsection{Convergence Guarantees}

\textbf{Trust-region stability.}
Local Newton-like preconditioning is fast once the market is close to consensus, but it can be unreliable far from consensus: the local curvature model may not describe the cost function well over a large step.
The trade cap \(s(r)\le\tau\) controls this risk, restricting each trade to a region where curvature stability keeps the local model accurate.
This gives the cap a trust-region interpretation: the local quadratic model is used only where it remains reliable~\cite{conn2000trust}.
Close to the target belief, the cap becomes inactive and the update behaves like a damped Newton step, so the same cap that ensures the market-design guarantees also stabilizes price discovery.

\begin{theorem}[Price discovery under the fixed-envelope fee]
\label{thm:ca-price-discovery}
Under the regularity assumptions in Appendix~\ref{app:preliminaries}, assume that \(C\) satisfies curvature stability, \(\tau>0\), and \(A=\kappa\tau\le A^\star\), where \(A^\star\) is the positive solution of \(e^A=1+2A\).
Assume also that the gauge is invariant to riskless shifts and restricts to a norm on \(E\).
Fix a belief \(\mu\in\operatorname{relint}(\Delta_d)\).
If a sequence of traders with belief \(\mu\) repeatedly maximizes expected profit under the fixed-envelope fee, then the induced prices \(p_t=\nabla C(q_t)\) satisfy
\begin{equation}
    p_t\to \mu .
\end{equation}
Moreover, once the iterates enter the local regime where the cap is inactive, the price error contracts geometrically:
\begin{equation}
    \|p_{t+1}-\mu\|
    \le
    \gamma_A
    \|p_t-\mu\|
    +
    O\bigl(\|p_t-\mu\|^2\bigr),
\end{equation}
where \(\gamma_A=|1-\eta_A|<1\) and \(\eta_A=1/(2\omega^+(A))\).
\end{theorem}

The proof is deferred to the appendix.
The theorem gives two guarantees for the fixed-envelope baseline: prices converge, then contract geometrically at the damping-dependent rate \(\gamma_A\).

SpanPM's trader objective is generally nonconvex because its multiplier depends on \(s(r)\). Nevertheless, after restricting \(r\) to zero-sum representatives in \(E\), a global response exists on the compact capped set, and the problem admits a useful reduction:
\begin{equation}
 \min_{0\le u\le\bar\tau}\;\min_{s(r)\le u}
 \{\langle p_q-\mu,r\rangle+\omega^+(\kappa u)\langle r,\Hq r\rangle\}.
 \label{eq:outer-span}
\end{equation}
For fixed \(u\), the inner problem is convex; for softmax it is solved exactly by weighted clipping. At a nonzero joint optimum one can take \(u=s(r)\), since reducing slack in \(u\) weakly lowers the multiplier. Thus only a bounded one-dimensional outer search remains.

\begin{theorem}[Price discovery under SpanPM]
\label{thm:span-price-discovery}
Under the regularity assumptions of Theorem~\ref{thm:ca-price-discovery}, let \(\kappa\bar\tau\le A^\star\). If traders with a common interior belief choose global best responses under SpanPM, then \(p_t\to\mu\). Once responses are interior,
\begin{equation}
 \|p_{t+1}-\mu\|=O(\|p_t-\mu\|^2).
\end{equation}
\end{theorem}
The rate follows from \(\omega^+(\kappa s(r))=1/2+O(\|r\|)\): the optimality condition becomes a full Newton step plus a second-order remainder. Unlike the fixed-envelope baseline, no constant damping remains near consensus. The global convergence and nonsmooth local expansion are proved in the appendix.
Appendix Figure~\ref{fig:price-discovery} illustrates the transition from cap-constrained movement to the predicted local regime.

\section{Experiments}
\label{sec:experiments}

We ask whether SpanPM (i) retains local-risk calibration, (ii) improves repeated aggregation under the \emph{same} hard safety cap as a fixed-envelope local baseline, and (iii) reduces charges when forecasts evolve with new signals.

\textbf{Common setup.}
We compare SpanPM with DCFMM~\cite{abernethy2013efficient}, SQPM~\cite{nueve2025smooth}, and the fixed-envelope baseline FixedEnv using softmax \((L=3)\) and fixed seeds. Both local mechanisms enforce \(\tau_{\max}=0.15\) and hence share the worst-case overcharge bound \(\rho(0.9)\approx1.825\); FixedEnv uses the fee multiplier \(\omega^+(0.9)\) for every trade, while SpanPM uses \(\omega^+(\kappa s(r))\). SpanPM responses combine exact weighted clipping with a 17-point multi-basin search and golden-section refinement, audited against a 1{,}001-point grid.

\subsection{Fee Calibration}
\label{sec:exp-calibration}

For each \(d\in\{2,5,10,20\}\), we draw up to 3{,}000 nondegenerate state--trade pairs from near-uniform through near-vertex prices. SQPM's fee-to-Bregman ratio exceeds 100 near simplex vertices, whereas FixedEnv remains below 1.60. SpanPM is pointwise no larger than FixedEnv and its factor approaches one for small trades. Appendix Table~\ref{tab:calibration} gives the full ratios; audited solver cases preserve domination and the drift inequality to floating-point tolerance.

\subsection{Repeated Forecast Aggregation}
\label{sec:exp-speed}

For each \(d\in\{3,5,10\}\), we run 200 paired trials with a target belief \(\mu\sim\mathrm{Dirichlet}(1.5\mathbf{1})\) and a confident but incorrect initial price.
A sequence of traders sharing \(\mu\) makes 20 numerically computed response trades, and we track \(\lVert p_t-\mu\rVert\).

\begin{figure}[t]
\centering
\includegraphics[width=0.88\columnwidth]{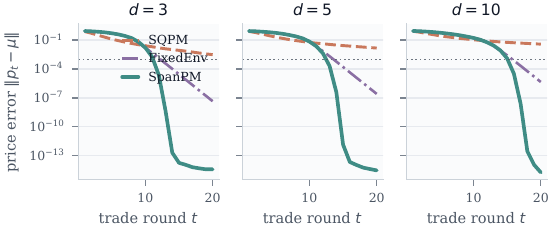}
\caption{Median price error over 200 paired confident-start trials. SpanPM and FixedEnv use the same hard cap and worst-case fee bound; the dotted line is tolerance \(10^{-3}\).}
\label{fig:exp-convergence}
\end{figure}

Figure~\ref{fig:exp-convergence} shows the local mechanisms overtaking SQPM and SpanPM entering its predicted quadratic regime after the cap becomes inactive. At round 20, FixedEnv's median errors are \(4.6\times10^{-8}\), \(2.6\times10^{-7}\), and \(4.5\times10^{-6}\) for \(d=3,5,10\), while all SpanPM medians are below \(10^{-12}\). SpanPM reaches tolerance \(10^{-3}\) in 100\%, 100\%, and 99.5\% (199/200) of trials, respectively, versus 100\%, 100\%, and 97.5\% for FixedEnv and 24.5\%, 0\%, and 0\% for SQPM; its median rounds are 12, 13, and 16. The cap sweep in Appendix~\ref{app:tau-sensitivity} rules out a larger-cap explanation.

\subsection{Evolving Synthetic Beliefs}
\label{sec:exp-informed}

For 100 paired seeds and \(d\in\{5,10\}\), 40 categorical signals update a common Dirichlet belief, and every mechanism replays the same path from the same confident, incorrect price. Final errors are similar because signal uncertainty limits accuracy. SQPM retains only 53--55\% of DCFMM surplus and accumulates \(0.338\)--\(0.406\) excess fee. SpanPM retains 96--97\% of DCFMM surplus, cuts excess fee by 94\% relative to SQPM, and lowers it by a further 10--13\% versus FixedEnv under the same cap. Appendix Table~\ref{tab:informed} reports the full means and metric definitions.

\section{Related Work}
\label{sec:related}

Financial language models produce probabilistic forecasts\newtext{~\cite{araci2019finbert,xie2023pixiu,lin2025cspo}} \newtext{and are benchmarked on multilingual financial question answering~\cite{xie2026finmmevaltask1,xie2026finmmevaltask2}}; SpanPM aggregates such outputs downstream. Prediction markets aggregate dispersed forecasts~\cite{wolfers2004prediction} using scoring-rule and convex-cost designs~\cite{hanson2003combinatorial,abernethy2013efficient,agrawal2011unified}, with expert-aggregation, message-passing, and online-learning interpretations~\cite{storkey2011machine,chen2010noregret,dellapenna2012interpreting}. Their links to proper scoring rules and Bregman divergences are well established~\cite{gneiting2007strictly,abernethy2012characterization,frongillo2018axiomatic}; our narrower rigidity result fixes the price-plus-fee rule.

Protection mechanisms study bounded loss, liquidity sensitivity, adaptive liquidity\newtext{, and recovery of prefunded seed capital through direction-conditioned fees}\newtext{~\cite{chen2007utility,othman2013practical,abernethy2014volume,nueve2026adaptive,chen2026seedcapital}}. SpanPM instead holds liquidity fixed and adapts a local-curvature fee to state and realized trade span, addressing SQPM's global calibration gap~\cite{nueve2025smooth}. Its dynamics complement first-order convergence analyses~\cite{frongillo2015convergence} and are analogous to Online Newton Step~\cite{hazan2007logarithmic}.

\section{Conclusion}

SpanPM aggregates probabilistic financial forecasts in a prediction market by pricing local curvature from each trade's realized spread while retaining a hard safety cap. Its calibration tightens as trades shrink and turns damped Newton aggregation into locally quadratic full-Newton convergence; simulations and search audits support these predictions while retaining near-DCFMM surplus.

\section*{Limitations}

Our guarantees concern smooth convex cost-function markets with curvature stability, a per-trade spread cap, and exact global responses. Experiments use a softmax market and synthetic beliefs, with a common trade-span cap for the two local mechanisms; they do not establish performance for deployed markets or financial NLP encoders. Traders are myopic and risk-neutral, excluding heterogeneous preferences, strategic splitting, latency, adversarial or correlated errors, and approximate or forward-looking responses. For softmax, the nonconvex response is handled by deterministic bounded search and a dense-grid audit, which is numerical rather than symbolic certification. The quadratic rate also assumes a fixed interior belief, nonsingular risky-subspace curvature, and an inactive cap, so evolving or heterogeneous forecasts need not attain it. Real-model evaluation and robust incentives remain future work.

\bibliography{references}

\appendix
% \section{Appendix}

\section{Additional Preliminaries}
\label{app:preliminaries}

This section collects standard technical assumptions, protocol details, and formal guarantee definitions that we omit from the Preliminaries section of the main text for space.

\textbf{Standard assumptions on the cost function.}
Following the standard cost-function market-making framework \cite{abernethy2013efficient} and the price-plus-fee formulation of \cite{nueve2025smooth}, we assume that \(C\) is convex, increasing, translation invariant, and a probability mapping.
Translation invariance means \(C(q+\alpha\mathbf{1})=C(q)+\alpha\) for all \(q\) and \(\alpha\).
The probability-mapping condition means that \(\nabla C(q)\in\Delta_d\) for every state \(q\), and that the closure of the price range \(\{\nabla C(q):q\in\mathbb{R}^d\}\) is the full simplex \(\Delta_d\).
The closure condition gives expressiveness: every belief in the simplex can be approximated by some market state.
For bounded-loss claims, we also assume \(C^*(\delta_i)<\infty\) for every outcome \(i\), where \(C^*\) is the Fenchel conjugate defined below.
The DCFMM's worst-case loss from \(q_0\) is then bounded by
\(\max_i\{C(q_0)-q_{0,i}+C^*(\delta_i)\}\).

Translation invariance also isolates the riskless direction of the market: the bundle \(\alpha\mathbf{1}\) pays exactly \(\alpha\) no matter which outcome occurs, so it should cost exactly \(\alpha\).
Differentiating \(C(q+\alpha\mathbf{1})=C(q)+\alpha\) gives \(\langle\nabla C(q),\mathbf{1}\rangle=1\), and differentiating again gives \(\nabla^2C(q)\mathbf{1}=0\).
Thus the Hessian is singular along \(\mathbf{1}\), and the meaningful local geometry of the market lies on the zero-sum subspace \(E=\{v\in\mathbb{R}^d:\langle \mathbf{1},v\rangle=0\}\).

\textbf{Regularity used in the proofs.}
For results involving inverse prices or Newton steps, we assume that \(C\in C^3\), \(\ker\nabla^2C(q)=\operatorname{span}\{\mathbf{1}\}\), and \(\nabla C|_E\) is a diffeomorphism from \(E\) to \(\operatorname{relint}(\Delta_d)\).
We also assume that \(\Phi_\mu(q)=C(q)-\langle\mu,q\rangle\) is coercive on \(E\) for every \(\mu\in\operatorname{relint}(\Delta_d)\).
These standard Legendre-type conditions hold for the softmax cost.
The price-discovery result additionally requires the trade-size gauge to be invariant under riskless shifts and to restrict to a norm on \(E\), as does the softmax spread gauge.

\textbf{Dual construction.}
Cost functions are often constructed by convex duality.
Let \(f^*(z)=\sup_x\{\langle z,x\rangle-f(x)\}\) denote the Fenchel conjugate.
Starting from a strictly convex regularizer \(\widehat C\) on \(\Delta_d\) or \(\operatorname{relint}(\Delta_d)\), one sets \(C=\widehat C^*\).
Under standard conditions, the resulting price map satisfies
\begin{equation}
    \nabla C(q)
    =
    \argmax_{p\in\operatorname{dom}(\widehat C)}
    \{\langle q,p\rangle-\widehat C(p)\},
\end{equation}
and its price range is dense in \(\Delta_d\) \cite{abernethy2013efficient}.
This duality is what underlies the duality-based cost-function market maker.

\textbf{Trading protocol.}
An automated market maker starts from an initial state \(q_0\).
At round \(t\), a trader requests a bundle \(r_t\), pays \(\mathrm{Pay}(q_t,r_t)\), and the state updates to \(q_{t+1}=q_t+r_t\).
After the outcome \(y\) is realized, the trader receives \(\langle r_t,\rho(y)\rangle\), where \(\rho(y)=\delta_i\) is the payoff vector when outcome \(y=i\) occurs, with \(\delta_i\in\mathbb{R}^d\) the \(i\)-th standard basis vector.
Positive coordinates of a bundle correspond to buying shares, while negative coordinates correspond to selling or shorting shares.

\textbf{Formal market-design guarantees.}
A mechanism should have well-defined instantaneous prices, rule out arbitrage, represent arbitrary beliefs, incorporate information through trading, and bound the market maker's worst-case loss.
No arbitrage means that, for every finite admissible history \((r_t)_{t=0}^{T-1}\), some outcome \(i\) satisfies \(\sum_t\mathrm{Pay}(q_t,r_t)\ge\sum_t(r_t)_i\).
Expressiveness means that every belief \(p\in\Delta_d\) can be approximated by some instantaneous price.
Information incorporation means that buying the same bundle twice should not become cheaper, i.e., \(\mathrm{Pay}(q+r,r)\ge \mathrm{Pay}(q,r)\).

The strongest incentive notion is one-shot incentive compatibility.
A trader with belief \(\mu\in\Delta_d\) chooses a trade maximizing expected profit,
\begin{equation}
    \max_r
    \left\{
        \langle \mu,r\rangle-\mathrm{Pay}(q,r)
    \right\}.
\end{equation}
One-shot incentive compatibility requires that such a trader can move the market price to their belief in a single trade: \(\mathrm{InstPrice}(q+r)=\mu\), where \(\mathrm{InstPrice}(q)=\nabla C(q)\).

The DCFMM satisfies all of the guarantees above, including one-shot incentive compatibility and bounded worst-case loss \cite{abernethy2013efficient}.
SQPM preserves the main safety and expressiveness guarantees but replaces one-shot incentive compatibility with an incremental form: if a sequence of traders share the same belief \(\mu\) and each maximizes expected profit, the induced prices satisfy \(\nabla C(q_t)\to\mu\) over time \cite{nueve2025smooth}.

\section{Additional Experimental Results}
\label{app:additional-experiments}

\begin{table}[h]
\centering
\small
\begin{tabular}{@{}rrrrr@{}}
\toprule
& \multicolumn{2}{c}{SQPM ratio} & \multicolumn{2}{c}{FixedEnv ratio} \\
\(d\) & mean & max & mean & max \\
\midrule
2  & \(>\!100\) & \(>\!100\) & 1.39 & 1.59 \\
5  & 22.7       & \(>\!100\) & 1.38 & 1.57 \\
10 & 15.6       & \(>\!100\) & 1.38 & 1.52 \\
20 & 24.8       & \(>\!100\) & 1.38 & 1.46 \\
\bottomrule
\end{tabular}
\caption{Fee-to-Bregman ratios over up to 3{,}000 nondegenerate state--trade pairs per \(d\). SpanPM is pointwise no larger than FixedEnv and has the sharper factor \(\rho(\kappa s(r))\).}
\label{tab:calibration}
\end{table}

\begin{table}[h]
\centering
\small
\begin{tabular}{@{}llrrr@{}}
\toprule
\(d\) & mechanism & err.\ @40 & surplus & excess fee \\
\midrule
     & DCFMM  & 0.112 & 0.754 & 0.000 \\
5    & SQPM   & 0.115 & 0.414 & 0.338 \\
     & FixedEnv & 0.112 & 0.726 & 0.023 \\
     & SpanPM & 0.112 & 0.729 & 0.020 \\
\addlinespace
     & DCFMM  & 0.121 & 0.908 & 0.000 \\
10   & SQPM   & 0.124 & 0.482 & 0.406 \\
     & FixedEnv & 0.121 & 0.868 & 0.026 \\
     & SpanPM & 0.121 & 0.871 & 0.024 \\
\bottomrule
\end{tabular}
\caption{Means after 40 signals over 100 paired seeds.}
\label{tab:informed}
\end{table}

The evolving-belief experiment reports final price error, cumulative ex-ante trader surplus, and cumulative excess fee
\(
\sum_t[\mathrm{Fee}(q_t,r_t)-D_C(q_t+r_t,q_t)]
\), the payment above exact exposure.

\begin{figure}[h]
\centering
\includegraphics[width=\columnwidth]{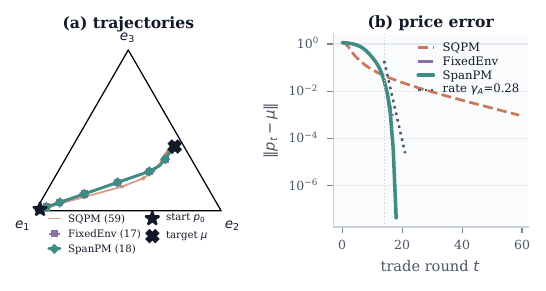}
\caption{Repeated numerically computed responses toward a fixed belief (softmax, \(L=3\)). \textbf{(a)} SQPM needs 59 trades to reach tolerance \(10^{-3}\), the fixed-envelope baseline needs 17, and SpanPM reaches the tighter \(10^{-6}\) tolerance in 18. \textbf{(b)} SpanPM matches capped movement far from consensus, then removes fixed damping as its realized trade span shrinks.}
\label{fig:price-discovery}
\end{figure}

\section{Illustrative Figure Details}
\label{app:figure-details}

Both Figure~\ref{fig:overcharge-1d} and Figure~\ref{fig:overcharge-simplex} use the softmax cost with \(L=3\) and a boundary trade \(s(r)=\tau=0.15\), giving \(A=\kappa\tau=0.9\), the cap shared by FixedEnv and SpanPM.

\subsection{Trade-Cap Sensitivity}
\label{app:tau-sensitivity}

Table~\ref{tab:tau-sensitivity} reports a paired sweep using the protocol of Section~\ref{sec:exp-speed}, with 200 seeds per dimension and \(T=20\). Larger caps accelerate early information incorporation but loosen the theoretical overcharge bound.

\begin{table}[h]
\centering
\small
\setlength{\tabcolsep}{3pt}
\begin{tabular}{@{}rrrrrr@{}}
\toprule
& & \multicolumn{2}{c}{\(d=5\)} & \multicolumn{2}{c}{\(d=10\)} \\
\(\tau\) & \(\rho(A)\) & err.\ @20 & conv. & err.\ @20 & conv. \\
\midrule
.025 & 1.105 & .710 & 0.0\% & .732 & 0.0\% \\
.050 & 1.221 & .375 & 0.0\% & .423 & 0.0\% \\
.100 & 1.493 & \(4.1{\times}10^{-5}\) & 69.5\% & .0157 & 31.0\% \\
.150 & 1.825 & \(2.6{\times}10^{-7}\) & 100.0\% & \(4.5{\times}10^{-6}\) & 97.5\% \\
.200 & 2.235 & \(2.6{\times}10^{-7}\) & 100.0\% & \(1.2{\times}10^{-6}\) & 100.0\% \\
\bottomrule
\end{tabular}
\caption{Median round-20 error and convergence rate at tolerance \(10^{-3}\) for FixedEnv under exact constrained best responses.}
\label{tab:tau-sensitivity}
\end{table}

\begin{figure}[h]
\centering
\includegraphics[width=.78\columnwidth]{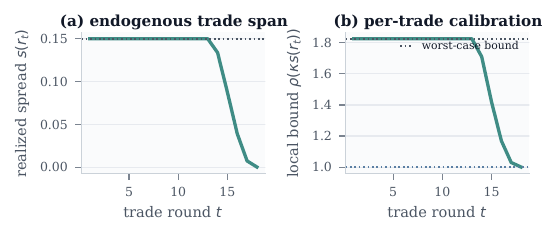}
\caption{SpanPM diagnostics for Figure~\ref{fig:price-discovery}. The realized span stays at the hard cap while disagreement is large, then shrinks toward zero; its trade-specific calibration factor simultaneously approaches one.}
\label{fig:span-diagnostics}
\end{figure}

\textbf{Response-solver audit.} On 20 held-out random states split across \(d=5,10\), the production 17-point multi-basin search was compared with a 1{,}001-point outer grid. Its largest objective gap was \(6.9\times10^{-18}\); the smallest fee-domination and drift margins were \(6.0\times10^{-4}\) and \(6.8\times10^{-3}\), respectively. These checks support the numerical implementation but do not constitute a symbolic global certificate.

On the illustrative fixed-belief path, regressing \(\log\|p_{t+1}-\mu\|\) on \(\log\|p_t-\mu\|\) over the five local points with error between \(10^{-5}\) and \(10^{-1}\) gives slope 1.84, approaching the quadratic-order prediction before numerical precision truncates the path.
Figure~\ref{fig:overcharge-1d} uses a two-outcome market (\(d=2\)) and sweeps the market's confidence in outcome~1 from parity toward certainty along a single axis.
Figure~\ref{fig:overcharge-simplex} extends this to a three-outcome market (\(d=3\)) and shows the same overcharge ratio over the full belief simplex, rather than along a single confidence axis, confirming that the pattern in Figure~\ref{fig:overcharge-1d} is not an artifact of the two-outcome case.

\section{Proofs for Curvature-Adaptive Fee Design}
\label{app:fee-design}

\subsection{Proof of Proposition~\ref{prop:bregman-unique}}

\noindent\textbf{Proposition~\ref{prop:bregman-unique}} (Uniqueness of the Bregman fee)\textbf{.} \textit{Fix a differentiable, translation-invariant cost function \(C\) whose price map is injective on \(E\). Consider a differentiable price-plus-fee market with payment \(\mathrm{Pay}(q,r)=\langle \nabla C(q),r\rangle+\mathrm{Fee}(q,r)\), where \(\mathrm{Fee}(q,0)=0\) and \(\mathrm{Fee}(q,r+\alpha\mathbf{1})=\mathrm{Fee}(q,r)\) for every \(\alpha\in\mathbb{R}\). If the market satisfies one-shot incentive compatibility for every belief in the price range of \(C\), then \(\mathrm{Fee}(q,r)=D_C(q+r,q)\).}

\begin{proof}
Fix \(q\) and an arbitrary \(r_E\in E\), and set \(\mu=\nabla C(q+r_E)\).
One-shot incentive compatibility provides an optimal trade whose terminal price is \(\mu\).
The objective is unchanged by adding \(\alpha\mathbf{1}\).
Writing \(q=q_E+\beta\mathbf{1}\), translation invariance of \(\nabla C\) and injectivity of \(\nabla C|_E\) therefore make the optimal trade's zero-sum component equal to \(r_E\).
The first-order condition on \(E\) is
\begin{equation}
    \nabla_E\mathrm{Fee}(q,r_E)
    =
    \mu-\nabla C(q)
    =
    \nabla_E D_C(q+r_E,q).
\end{equation}
Because \(r_E\) was arbitrary, \(\mathrm{Fee}(q,r_E)-D_C(q+r_E,q)\) is constant on \(E\).
Both terms vanish at \(r_E=0\), so the constant is zero.
Finally, translation invariance gives \(D_C(q+r_E+\alpha\mathbf{1},q)=D_C(q+r_E,q)\); the assumed invariance of \(\mathrm{Fee}\) extends the identity to every \(r\in\mathbb{R}^d\).
\end{proof}

\subsection{Proof of Theorem~\ref{thm:span-price-discovery}}

\begin{proof}
Work on the zero-sum subspace \(E\), and write
\(\Phi_\mu(q)=C(q)-\langle\mu,q\rangle\) and
\(e(q)=\nabla C(q)-\mu\). Let
\begin{equation}
 m_q(r)=\langle e(q),r\rangle+
 \mathrm{Fee}_{\mathrm{span}}(q,r)
\end{equation}
on the compact ball \(B=\{r\in E:s(r)\le\bar\tau\}\). The objective is continuous, so a global minimizer exists; choose any one and call it \(r_q\). Fee domination gives
\begin{equation}
 \Phi_\mu(q+r_q)-\Phi_\mu(q)
 \le m_q(r_q).
 \label{eq:span-descent}
\end{equation}

The sublevel set \(K=\{q\in E:\Phi_\mu(q)\le\Phi_\mu(q_0)\}\) is compact. On \(K\), \(\|\nabla^2C(q)\|\le M\), and
\(\omega^+(\kappa s(r))\le\omega^+(\bar A)=w_{\max}\). For a sufficiently small constant \(\delta>0\), the candidate \(r=-\delta e(q)\) is feasible throughout \(K\) and satisfies
\begin{equation}
\begin{aligned}
 m_q(-\delta e(q))
 &\le(-\delta+w_{\max}M\delta^2)\|e(q)\|^2\\
 &\le-\tfrac{\delta}{2}\|e(q)\|^2.
\end{aligned}
\end{equation}
Because \(r_q\) is global, it does at least as well. Away from the unique \(q^\star\) with \(\nabla C(q^\star)=\mu\), \(\|e(q)\|\) has a positive minimum on compact subsets of \(K\). Equation~\eqref{eq:span-descent} therefore gives a uniform Lyapunov decrease away from every neighborhood of \(q^\star\). Since \(\Phi_\mu\) is bounded below, \(q_t\to q^\star\), hence \(p_t\to\mu\).

For the local rate, uniformly near \(q^\star\),
\begin{align}
 \mathrm{Fee}_{\mathrm{span}}(q,r)
 &=\tfrac12 r^\top\nabla^2C(q)r+R_q(r),\\
 R_q(r)&=O(\|r\|^3),\quad
 \partial R_q(r)=O(\|r\|^2),
\end{align}
where \(\partial\) is the Clarke subdifferential. The second relation follows because \(s\) is Lipschitz on \(E\), \(\omega^+(a)=1/2+a/6+O(a^2)\), and the quadratic form and its gradient have orders two and one. At \(e=0\), the unique global response is zero; compactness therefore forces global responses to remain local as \(e\to0\). Local positive definiteness then gives \(\|r_q\|=O(\|e(q)\|)\). Clarke stationarity yields
\begin{equation}
 0\in e(q)+\nabla^2C(q)r_q+O(\|r_q\|^2),
\end{equation}
so
\(r_q=-\nabla^2C(q)^\dagger e(q)+O(\|e(q)\|^2)\).
The cap is inactive locally, and a Taylor expansion gives
\begin{align}
 e(q+r_q)
 &=e(q)+\nabla^2C(q)r_q+O(\|r_q\|^2)\\
 &=O(\|e(q)\|^2),
\end{align}
which is the stated quadratic price-error rate.
\end{proof}

\subsection{Proof of Theorem~\ref{thm:dominating-guarantees}}

\noindent\textbf{Theorem~\ref{thm:dominating-guarantees}} (Guarantees from dominating fees)\textbf{.} \textit{If \(\mathrm{Fee}\) is \(C\)-dominating, then the corresponding price-plus-fee market has well-defined instantaneous prices, satisfies no arbitrage, preserves expressiveness, and incorporates information. Moreover, on every trade history, the market maker's collected revenue is at least the revenue of the corresponding DCFMM. Consequently, the market maker's worst-case loss is no larger than under the DCFMM.}

\begin{proof}
By~\eqref{eq:dominating-price},
\(\mathrm{Pay}(q,\varepsilon v)/\varepsilon\to\langle\nabla C(q),v\rangle\), so the instantaneous price remains \(\nabla C(q)\) and expressiveness is unchanged.

Condition~\eqref{eq:dominating-protection} gives
\begin{align}
    \mathrm{Pay}(q,r)
    &\ge \langle p_q,r\rangle+D_C(q+r,q) \notag\\
    &=C(q+r)-C(q).
\end{align}
Hence, along any finite admissible history,
\begin{equation}
    \sum_{t=0}^{T-1}\mathrm{Pay}(q_t,r_t)
    \ge C(q_T)-C(q_0).
\end{equation}
Let \(m=\min_i(q_T-q_0)_i\) and choose \(i\) attaining the minimum.
Since \(q_T\ge q_0+m\mathbf{1}\) coordinatewise, monotonicity and translation invariance imply
\(C(q_T)-C(q_0)\ge m=\sum_t(r_t)_i\), which rules out a risk-free profit.

For information incorporation, let
\(\Delta\mathrm{Pay}=\mathrm{Pay}(q+r,r)-\mathrm{Pay}(q,r)\).
\begin{align}
    \Delta\mathrm{Pay}
    &=\langle p_{q+r}-p_q,r\rangle \notag\\
    &\quad+
    \mathrm{Fee}(q+r,r)-\mathrm{Fee}(q,r).
\end{align}
Condition~\eqref{eq:dominating-drift} makes this quantity nonnegative, so \(\mathrm{Pay}(q+r,r)\ge \mathrm{Pay}(q,r)\).

The same pathwise inequality shows that total revenue is at least the DCFMM revenue \(C(q_T)-C(q_0)\).
Realized payouts are identical, so the loss is no larger for any history and outcome, and therefore no larger in the worst case.
\end{proof}

\subsection{Proof of Proposition~\ref{prop:sqpm-unbounded-overcharge}}

\noindent\textbf{Proposition~\ref{prop:sqpm-unbounded-overcharge}} (Unbounded overcharge of the global quadratic fee)\textbf{.} \textit{For the softmax cost function, there exist fixed nonzero trades \(r\) such that the overcharge ratio is unbounded: for every \(M>0\), there is a market state \(q\) at which \(\mathrm{Fee}_L(q,r)/D_C(q+r,q) > M\).}

\begin{proof}
For the softmax cost \(C(q)=\frac{1}{L}\log\sum_{i=1}^d e^{Lq_i}\), the Bregman fee can be written as
\begin{align}
    D_C(q+r,q)
    &=
    \frac{1}{L}\log\left(\sum_i (p_q)_i e^{Lr_i}\right) \notag\\
    &\quad-\sum_i (p_q)_i r_i.
\end{align}
Choose a fixed nonconstant trade \(r\notin\operatorname{span}\{\mathbf{1}\}\).
Fix an outcome \(k\) and let \(q^{(n)}=n\delta_k\), so that \((p_{q^{(n)}})_k=e^{Ln}/(e^{Ln}+(d-1))\to 1\) and \((p_{q^{(n)}})_i\to 0\) for \(i\ne k\), i.e., \(p_{q^{(n)}}\to \delta_k\).
Then
\begin{equation}
    D_C(q^{(n)}+r,q^{(n)})
    \to
    \frac{1}{L}\log(e^{Lr_k})-r_k
    =
    0.
\end{equation}
For every finite \(n\), the softmax distribution has full support, so strict Jensen inequality gives \(D_C(q^{(n)}+r,q^{(n)})>0\).
By contrast, \(\mathrm{Fee}_L(q^{(n)},r)=\frac{L}{2}\|r\|^2\) is constant and strictly positive.
Therefore
\begin{equation}
    \frac{\mathrm{Fee}_L(q^{(n)},r)}
    {D_C(q^{(n)}+r,q^{(n)})}
    \to
    \infty
\end{equation}
as \(n\to\infty\), which proves the claim.
\end{proof}

\subsection{Guarantees for the Curvature-Adaptive Fee}

\begin{theorem}[Safety and bounded overcharge]
\label{thm:ca-safety-overcharge}
Assume that \(C\) satisfies \(\kappa\)-curvature stability, and restrict trades to \(s(r)\le\tau\).
Let \(A=\kappa\tau\).
Then, for every admissible trade,
\begin{equation}
    D_C(q+r,q)
    \le
    \mathrm{Fee}_A(q,r)
    \le
    \rho(A)D_C(q+r,q),
\end{equation}
where \(\rho(A)=\omega^+(A)/\omega^-(A)\).
Thus the overcharge ratio is at most \(\rho(A)\) whenever \(D_C(q+r,q)>0\); if \(D_C(q+r,q)=0\), both fees vanish.
\end{theorem}

The first inequality gives market-maker protection, and the second gives risk calibration.
Unlike the SQPM fee, the fixed-envelope fee's overcharge ratio cannot diverge in confident markets.

\begin{figure}[t]
\centering
\includegraphics[width=\columnwidth]{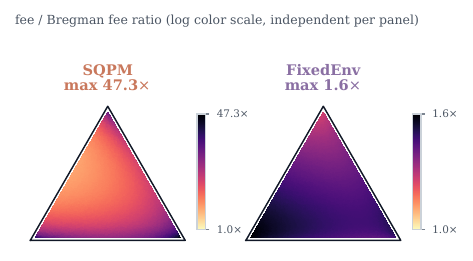}
\caption{Overcharge ratio \(\mathrm{Fee}(q,r)/D_C(q+r,q)\) over the belief simplex \(\Delta_3\) for a fixed trade \(r\) with \(s(r)=\tau\) (softmax cost, \(L=3\), \(A=0.9\), matching Figures~\ref{fig:overcharge-1d} and~\ref{fig:price-discovery}). Darker means larger overcharge; \textbf{each panel uses its own log color scale}, since the two ranges differ by more than an order of magnitude and a shared scale would make FixedEnv's variation invisible. \textbf{Left:} SQPM's ratio grows without bound as the market approaches the vertex where \(r\) carries little risk (Proposition~\ref{prop:sqpm-unbounded-overcharge}), reaching \(47\times\) here. \textbf{Right:} FixedEnv's ratio stays within the uniform bound \(\rho(A)\approx 1.83\) guaranteed by Theorem~\ref{thm:ca-safety-overcharge}, topping out at \(1.6\times\), roughly \(30\times\) tighter than SQPM at its worst.}
\label{fig:overcharge-simplex}
\end{figure}

For the softmax cost, local curvature has the form \(\langle r,\nabla^2 C(q)r\rangle=L\,\operatorname{Var}_{p_q}(r)\), so the fixed-envelope fee charges in proportion to the payoff variance of the trade under the current market belief: trades with high payoff uncertainty are charged more, while trades that are nearly deterministic under the current belief are charged less.

\begin{theorem}[Information incorporation]
\label{thm:ca-information-incorporation}
Assume that \(C\) satisfies \(\kappa\)-curvature stability, restrict trades to \(s(r)\le\tau\), and let \(A=\kappa\tau\).
If \(A\le A^\star\), where \(A^\star\) is the positive solution of
\begin{equation}
    e^A=1+2A,
\end{equation}
then, for every admissible trade, the curvature-adaptive fee satisfies the drift condition
\begin{equation}
    \mathrm{Fee}_A(q,r)-\mathrm{Fee}_A(q+r,r)
    \le
    \langle p_{q+r}-p_q,r\rangle .
\end{equation}
Consequently, the fixed-envelope rule preserves information incorporation.
Combined with Theorem~\ref{thm:ca-safety-overcharge}, it also preserves no arbitrage, expressiveness, and bounded worst-case loss no larger than the corresponding DCFMM.
\end{theorem}

\subsection{Proof of Theorem~\ref{thm:ca-safety-overcharge}}

\begin{proof}
Fix an admissible \(r\) and let \(a=\kappa s(r)\le A\).
The local two-sided bound gives
\begin{equation}
    \begin{aligned}
        \omega^-(a)\langle r,\nabla^2 C(q)r\rangle
        &\le D_C(q+r,q),\\
        D_C(q+r,q)
        &\le \omega^+(a)\langle r,\nabla^2 C(q)r\rangle .
    \end{aligned}
\end{equation}
Since \(\omega^+\) is nondecreasing and \(\omega^-\) is nonincreasing,
\begin{align}
    D_C(q+r,q)
    &\le \omega^+(A)\langle r,\Hq r\rangle \notag\\
    &=\mathrm{Fee}_A(q,r)
\end{align}
and
\begin{align}
    D_C(q+r,q)
    &\ge \omega^-(A)\langle r,\Hq r\rangle \notag\\
    &=\frac{\omega^-(A)}{\omega^+(A)}\mathrm{Fee}_A(q,r).
\end{align}
Rearranging the second inequality yields the stated factor \(\rho(A)\).
\end{proof}

\subsection{Proof of Theorem~\ref{thm:ca-information-incorporation}}

\begin{proof}
Let \(g(\theta)=\langle r,\nabla^2 C(q+\theta r)r\rangle\ge0\) and \(a=\kappa s(r)\le A\).
The two sides of the drift condition are, respectively,
\(\omega^+(A)[g(0)-g(1)]\) and
\begin{equation}
    \langle \nabla C(q+r)-\nabla C(q),r\rangle
    =
    \int_0^1
    g(\theta)
    \,d\theta.
\end{equation}
If \(a=0\), curvature stability makes \(g\) constant, and the inequality is immediate.
If \(g(0)=0\), the upper curvature bound gives \(g(\theta)=0\) throughout.
It remains to consider \(a>0\) and \(g(0)>0\).
Curvature stability gives \(g(\theta)\ge e^{-a\theta}g(0)\) for every \(\theta\in[0,1]\) (in particular \(g(1)\ge e^{-a}g(0)\)), so
\begin{align*}
    g(0)-g(1)
    &\le
    g(0)\bigl(1-e^{-a}\bigr), \\
    \int_0^1 g(\theta)\,d\theta
    &\ge
    g(0)\int_0^1 e^{-a\theta}\,d\theta
    =
    g(0)\,\frac{1-e^{-a}}{a}.
\end{align*}
Thus it suffices that \(\omega^+(A)\le1/a\) for every \(a\in(0,A]\).
Since \(a>0\) implies \(A>0\), the tightest requirement is
\begin{align}
    \omega^+(A)\le\frac{1}{A}
    &\iff
    \frac{e^A-1-A}{A^2}\le\frac{1}{A} \notag\\
    &\iff
    e^A\le 1+2A,
\end{align}
which holds for \(A\le A^\star\).
Thus~\eqref{eq:dominating-drift} holds.
Moreover, \(\mathrm{Fee}_A(q,r)=O(\|r\|^2)=o(\|r\|)\), and Theorem~\ref{thm:ca-safety-overcharge} supplies~\eqref{eq:dominating-protection}.
The remaining guarantees follow from Theorem~\ref{thm:dominating-guarantees}.
\end{proof}

\subsection{Proof of Theorem~\ref{thm:span-guarantees}}

\begin{proof}
Fix \(r\) and write \(a=\kappa s(r)\le\bar A\),
\(g_q(r)=\langle r,\nabla^2C(q)r\rangle\), and
\(D=D_C(q+r,q)\). The local two-sided bound gives
\begin{equation}
 \omega^-(a)g_q(r)\le D
 \le\omega^+(a)g_q(r)=\mathrm{Fee}_{\mathrm{span}}(q,r).
\end{equation}
Rearranging the lower bound yields
\(\mathrm{Fee}_{\mathrm{span}}(q,r)\le\rho(a)D_C(q+r,q)\), and monotonicity of \(\rho\) gives the final inequality.

For information incorporation, let
\(g(\theta)=\langle r,\nabla^2C(q+\theta r)r\rangle\) and
\(\Delta_F=\mathrm{Fee}_{\mathrm{span}}(q,r)-\mathrm{Fee}_{\mathrm{span}}(q+r,r)\).
Because \(a\) depends on \(r\), not on the state,
\begin{equation}
\begin{aligned}
 \Delta_F&=\omega^+(a)[g(0)-g(1)]\\
 &\le\int_0^1g(\theta)\,d\theta
 =\langle p_{q+r}-p_q,r\rangle .
\end{aligned}
\end{equation}
The inequality is exactly the proof for fixed \(A\), now with \(A=a\le A^\star\). Finally, \(a=O(\|r\|)\), \(\omega^+(a)=1/2+O(\|r\|)\), and \(g_q(r)=O(\|r\|^2)\), so the fee is \(o(\|r\|)\). It is therefore \(C\)-dominating, and Theorem~\ref{thm:dominating-guarantees} supplies the remaining claims.
\end{proof}

\subsection{Global-response reduction}

Taking zero-sum representatives, let \(b=p_q-\mu\), \(H=\nabla^2C(q)\), and define the original response value
\begin{equation}
 \min_{r\in E:\,s(r)\le\bar\tau}
 \{\langle b,r\rangle+\omega^+(\kappa s(r))r^\top Hr\}.
 \label{eq:span-original-app}
\end{equation}
This equals the nested problem in~\eqref{eq:outer-span}. Indeed, every \(r\) in~\eqref{eq:span-original-app} is feasible in the nested problem with \(u=s(r)\), giving one direction. Conversely, for every feasible \((u,r)\), monotonicity of \(\omega^+\) and \(s(r)\le u\) imply
\begin{equation}
 \langle b,r\rangle+\omega^+(\kappa s(r))r^\top Hr
 \le \langle b,r\rangle+\omega^+(\kappa u)r^\top Hr.
\end{equation}
Thus a joint optimum may be chosen with \(u=s(r)\). For fixed \(u\), the inner objective is convex quadratic on a compact convex set. Its value is continuous in \(u\) by the maximum theorem, so a global outer minimizer exists. The softmax implementation evaluates a deterministic outer grid, refines every grid-visible basin, and checks the result against a much denser audit grid; because the outer value need not be convex, the experiments do not treat a single local scalar solve as a global response.

\section{Curvature Stability and Local Bounds}
\label{app:curvature}

\subsection{Softmax Curvature Stability}

\begin{proof}
For the softmax cost, the price vector is \(p_q=\nabla C(q)\), and
\begin{align*}
    C(q)&=\frac{1}{L}\log\sum_{i=1}^d e^{Lq_i}, \\
    \nabla^2 C(q)
    &=
    L\bigl(\operatorname{diag}(p_q)-p_qp_q^\top\bigr), \\
    \langle r,\nabla^2 C(q)r\rangle
    &=
    L\,\operatorname{Var}_{p_q}(r)
    \quad\text{for any }r.
\end{align*}
Under a state shift \(u\), the softmax weights are reweighted by factors \(e^{Lu_i}\).
Writing \(p=p_q\), \(p'=p_{q+u}\), and \(s(u)=\max_i u_i-\min_i u_i\),
\begin{equation}
    e^{-Ls(u)}
    \le \frac{p'_i}{p_i}
    =\frac{e^{Lu_i}}{\sum_j p_j e^{Lu_j}}
    \le e^{Ls(u)}.
\end{equation}
The pairwise identity
\begin{equation}
    \operatorname{Var}_p(r)
    =\frac12\sum_{i,j}p_ip_j(r_i-r_j)^2
\end{equation}
therefore yields
\begin{equation}
    e^{-2Ls(u)}\operatorname{Var}_{p}(r)
    \le \operatorname{Var}_{p'}(r)
    \le e^{2Ls(u)}\operatorname{Var}_{p}(r).
\end{equation}
Since this holds for every \(r\),
\begin{equation}
    \begin{aligned}
        e^{-2Ls(u)}\Hq
        &\preceq \nabla^2 C(q+u),\\
        \nabla^2 C(q+u)
        &\preceq e^{2Ls(u)}\Hq.
    \end{aligned}
\end{equation}
Hence softmax satisfies curvature stability with \(\kappa=2L\).
\end{proof}

\subsection{Proof of the Local Two-Sided Bound}

\begin{proof}
By the curvature integral representation, \(D_C(q+r,q)=\int_0^1(1-\theta)\langle r,\nabla^2 C(q+\theta r)r\rangle\,d\theta\); curvature stability gives
\begin{align}
    e^{-\kappa s(\theta r)}
    \langle r,\nabla^2 C(q)r\rangle
    &\le
    \langle r,\nabla^2 C(q+\theta r)r\rangle \notag\\
    &\le
    e^{\kappa s(\theta r)}
    \langle r,\nabla^2 C(q)r\rangle.
\end{align}
Since the gauge \(s\) is positively homogeneous, \(s(\theta r)=\theta\,s(r)\) for \(\theta\in[0,1]\).
Let \(a=\kappa s(r)\).
Then
\begin{align}
    D_C(q+r,q)
    &\le
    \langle r,\nabla^2 C(q)r\rangle
    \int_0^1 (1-\theta)e^{a\theta}\,d\theta, \notag\\
    D_C(q+r,q)
    &\ge
    \langle r,\nabla^2 C(q)r\rangle
    \int_0^1 (1-\theta)e^{-a\theta}\,d\theta.
\end{align}
Evaluating the integrals gives
\begin{align*}
    \int_0^1 (1-\theta)e^{a\theta}\,d\theta
    &=
    \frac{e^a-1-a}{a^2}
    =
    \omega^+(a), \\
    \int_0^1 (1-\theta)e^{-a\theta}\,d\theta
    &=
    \frac{e^{-a}-1+a}{a^2}
    =
    \omega^-(a).
\end{align*}
The formulas extend to \(a=0\) by continuity, with value \(1/2\).
\end{proof}

\begin{proof}[Monotonicity of \(\omega^+\) and \(\omega^-\)]
For \(a\ge0\), the integral representations give
\begin{align}
    \frac{d}{da}\omega^+(a)
    &=
    \int_0^1 (1-\theta)\theta\, e^{a\theta}\,d\theta
    \;\ge\; 0,
    \notag \\[2pt]
    \frac{d}{da}\omega^-(a)
    &=
    -\int_0^1 (1-\theta)\theta\, e^{-a\theta}\,d\theta
    \;\le\; 0,
\end{align}
so \(\omega^+\) is nondecreasing and \(\omega^-\) is nonincreasing.
In particular, \(\omega^+(A)\ge\omega^+(0)=1/2\).
\end{proof}

\section{Proofs for Price Discovery Dynamics}
\label{app:dynamics}

\subsection{SQPM as Gradient Descent}

\begin{proof}
Let \(\Phi_\mu(q)=C(q)-\langle \mu,q\rangle\).
The SQPM first-order condition gives
\(r=-(1/L)\nabla\Phi_\mu(q)\), hence
\begin{equation}
    q^+
    =
    q+r
    =
    q-\frac{1}{L}\nabla\Phi_\mu(q).
\end{equation}
\end{proof}

\subsection{Curvature-Adaptive Fees as Damped Newton Steps}

\begin{proof}
When the cap is inactive, the first-order condition on \(E\) is
\begin{equation}
    \nabla C(q)-\mu+2\omega^+(A)\nabla^2 C(q)r=0.
\end{equation}
The Hessian is positive definite on \(E\), so the zero-sum representative is
\begin{equation}
    r
    =
    \frac{1}{2\omega^+(A)}
    \nabla^2 C(q)^\dagger
    \bigl(\mu-\nabla C(q)\bigr).
\end{equation}
This is a Newton step for \(\Phi_\mu\) with damping \(\eta_A=1/(2\omega^+(A))\).
\end{proof}

\subsection{Proof of Theorem~\ref{thm:ca-price-discovery}}

\noindent\textbf{Theorem~\ref{thm:ca-price-discovery}} (Price discovery under the fixed-envelope fee)\textbf{.} \textit{Under the regularity assumptions in Appendix~\ref{app:preliminaries}, assume that \(C\) satisfies curvature stability, \(\tau>0\), and \(A=\kappa\tau\le A^\star\). Assume also that the gauge is invariant to riskless shifts and restricts to a norm on \(E\). Fix \(\mu\in\operatorname{relint}(\Delta_d)\). If traders with belief \(\mu\) repeatedly maximize expected profit under the fixed-envelope fee, then \(\nabla C(q_t)\to\mu\). Once the cap is inactive, \(\|\nabla C(q_{t+1})-\mu\|\le\gamma_A\|\nabla C(q_t)-\mu\|+O(\|\nabla C(q_t)-\mu\|^2)\), where \(\gamma_A=|1-\eta_A|<1\) and \(\eta_A=1/(2\omega^+(A))\).}

\begin{proof}
Because both the objective and the cap are invariant under riskless shifts, take zero-sum representatives throughout.
Write \(\Phi_\mu(q)=C(q)-\langle\mu,q\rangle\), \(e(q)=\nabla C(q)-\mu\), and \(w=\omega^+(A)\).
On the compact set \(B=\{r\in E:s(r)\le\tau\}\), the trader uniquely minimizes
\begin{equation}
    m_q(r)=\langle e(q),r\rangle+w\langle r,\nabla^2C(q)r\rangle.
\end{equation}
Uniqueness follows because \(\nabla^2C(q)\) is positive definite on \(E\).
Denote the minimizer by \(r(q)\).
If \(e(q)\ne0\), then \(m_q(-\varepsilon e(q))<0\) for all sufficiently small \(\varepsilon>0\), whereas \(m_q(0)=0\).
Thus \(m_q(r(q))<0\) unless \(q=q^\star\), the unique point in \(E\) satisfying \(\nabla C(q^\star)=\mu\).

Write \(r_q=r(q)\). The safety bound gives
\begin{align}
    &\Phi_\mu(q+r_q)-\Phi_\mu(q) \notag\\
    &\quad=\langle e(q),r_q\rangle+D_C(q+r_q,q) \notag\\
    &\quad\le m_q(r_q).
\end{align}
Hence \(\Phi_\mu\) decreases strictly away from \(q^\star\).
Let \(K=\{q\in E:\Phi_\mu(q)\le\Phi_\mu(q_0)\}\), which is compact by coercivity and contains every iterate.
The unique minimizer \(r(q)\) is continuous in \(q\), so
\(\Delta(q)=\Phi_\mu(q)-\Phi_\mu(q+r(q))\) is continuous and positive on \(K\setminus\{q^\star\}\).
For every \(\varepsilon>0\), compactness gives
\begin{equation}
    \inf_{q\in K:\,\|q-q^\star\|\ge\varepsilon}\Delta(q)>0.
\end{equation}
The iterates cannot enter this set infinitely often because \(\Phi_\mu(q_t)\) is bounded below.
Therefore \(q_t\to q^\star\), and continuity of \(\nabla C\) gives \(\nabla C(q_t)\to\mu\).

For the local rate, the unconstrained minimizer on \(E\) is
\begin{equation}
    \bar r(q)=-\eta_A\nabla^2C(q)^\dagger e(q),
    \qquad
    \eta_A=\frac{1}{2w}.
\end{equation}
It is continuous and vanishes at \(q^\star\); since \(\tau>0\), the cap is inactive in a neighborhood of \(q^\star\).
There, \(e_t=e(q_t)\in E\).
Since \(C\in C^3\) and the inverse Hessian on \(E\) is locally bounded, Taylor expansion gives
\begin{align}
    e_{t+1}
    &=e_t-\eta_A\nabla^2C(q_t)\nabla^2C(q_t)^\dagger e_t
      +O(\|e_t\|^2) \notag\\
    &=(1-\eta_A)e_t+O(\|e_t\|^2).
\end{align}
Taking norms yields the stated bound.
Finally, \(w\ge1/2\), so \(\eta_A\in(0,1]\) and \(\gamma_A=|1-\eta_A|<1\).
\end{proof}

\end{document}